\documentclass{article}
\usepackage{arxiv}
\usepackage{authblk}

\usepackage[english]{babel}
\usepackage{multicol}
\usepackage{url}
\usepackage{amssymb}
\usepackage{amsfonts}
\usepackage{amsmath}
\usepackage{amsthm}
\usepackage{subfigure}
\usepackage{epsfig}
\usepackage{graphics}
\usepackage{dsfont}
\usepackage{color}
\usepackage{verbatim}
\usepackage{latexsym}
\usepackage{booktabs}
\usepackage{multirow}
\usepackage{setspace}
\usepackage{float}
\usepackage{booktabs}
\usepackage{placeins}
\usepackage{bbm}
\usepackage{graphicx}
\usepackage{eurosym}
\usepackage{xcolor}

\newtheorem{theorem}{Theorem}[section]
\newtheorem{proposition}[theorem]{Proposition}
\newtheorem{lemma}[theorem]{Lemma}
\newtheorem*{proposition*}{Proposition}

\newtheorem{definition}[theorem]{Definition}

\theoremstyle{definition}

\newtheorem*{remark}{Remark}

\usepackage{setspace}
\usepackage{todonotes}
\newcommand{\todoi}[2][] {\vspace{0.5em}\todo[inline, #1]{#2}}
\newboolean{includeHidden}
\setboolean{includeHidden}{false}
\newcommand{\hide}[1]{\ifthenelse{\boolean{includeHidden}}{{\tiny\textbf{HIDDEN:~}#1}}{}}

\allowdisplaybreaks[2]
 
\usepackage{rotating} 
 \usepackage{natbib}
 \bibpunct[, ]{(}{)}{,}{a}{}{,}%
 \def\newblock{\ }%
\usepackage{multirow}
\usepackage[caption=false]{subfig}

\usepackage[]{todonotes}
  \usepackage{setspace}
\usepackage{xpatch}

\newcommand{\x}{\mathbf{x}}
\newcommand{\y}{\mathbf{y}}
\newcommand{\p}{\mathbf{p}}
\newcommand{\bbf}{\mathbf{b}}

\newcommand{\e}{\mathbf{e}}

\DeclareMathOperator{\R}{\mathbb{R}}

\DeclareMathOperator{\conv}{conv}

\usepackage{amsmath}
\DeclareMathOperator*{\argmax}{arg\,max}

\usepackage{tikz}
\usetikzlibrary{arrows}
\usetikzlibrary{calc,fit}
\usetikzlibrary{positioning,chains}

\makeatletter

\def\tmk@labeldef#1,#2\@nil{%
  \def\tmk@label{#1}%
  \def\tmk@def{#2}%
}

\makeatother

\usepackage[vlined, ruled, linesnumbered]{algorithm2e}

\begin{document}

\title{Computing Candidate Prices in Budget-Constrained Product-Mix Auctions}

\author{Maximilian Fichtl}
\affil{Decision Sciences and Systems (DSS)\\
	Department of Informatics\\
	Technical University of Munich\\
	\texttt{fichtlm@in.tum.de}}

\maketitle

\begin{abstract}
	We study the problem of computing optimal prices for a version of the Product-Mix auction with budget constraints. In contrast to the ``standard'' Product-Mix auction, the objective is to maximize revenue instead of social welfare. We prove correctness of an algorithm proposed by Paul Klemperer and DotEcon which is sufficiently efficient in smaller markets.
\end{abstract}

\section{Introduction}
In this paper we consider the arctic version of the Product-Mix auction which was designed by Paul Klemperer for the Icelandic's government's strategy for exiting capital controls \citep{klemperer2010product}. The name ``arctic'' is chosen in contrast to the maybe more familiar version of the Product-Mix auction with quasi-linear utilities that can be analyzed with methods from tropical geometry \citep{baldwin2019understanding}. In contrast to many other auction formats, it allows bidders to express budget constraints. Additionally, the supply of goods is not assumed to be fixed, but depends on the seller's cost of providing a certain amount goods.

There are $n$ different goods $i \in \{1,\dots,n\}$ and a finite number of bidders, each submitting exactly one bid. A bid is a vector $\bbf = (b_1,\dots,b_n) \in \R^n_{\geq 0}$ together with a budget $\beta(\bbf) > 0$. Here, $b_i$ represents the value of bid $\bbf$ for one unit of good $i$. We identify bidders by their unique bid $\bbf$ and denote the set of all bids by $\mathcal{B}$. We assume that for each good $i \in \{1,\dots,n\}$, there is at least one bid with $b_i > 0$ -- otherwise we could just exclude good $i$ from the auction. A bundle is a vector $\x = (x_1,\dots,x_n) \in \R^n_{\geq 0}$ where $x_i$ denotes the possibly fractional amount of good $i$ contained in $\x$. An allocation is an assignment of bundles to bidders, where $\x^{\bbf}$ denotes the bundle assigned to bidder $\bbf \in \mathcal{B}$. The auctioneer selects a linear and anonymous price vector $\p = (p_1,\dots,p_n) \in \R^n_{> 0}$ and determines an allocation $\{\x^{\bbf}\}_{\bbf \in \mathcal{B}}$. Note that all coordinates of the price vector are strictly positive - a price of $0$ for some good would cause an infinite amount of that good to be demanded.  We say that an allocation is feasible, if $\sum_{i=1}^n p_i x^{\bbf}_i \leq \beta(\bbf)$ for all bidders $\bbf$, i.e., the cost of the assigned bundle does not exceed the available budget. A price vector $\p$ together with a feasible allocation $\{\x^{\bbf}\}_{\bbf \in \mathcal{B}}$ generates the utility
\[
\sum_{i=1}^n b_i x^{\bbf}_i - \sum_{i=1}^n p_i x^{\bbf}_i = \langle \bbf - \p, \x^{\bbf} \rangle
\]
for bidder $\bbf$. Given a price vector $\p$, the demand set of bidder $\bbf$ is the set of all bundles $\x$ maximizing their utility: 
\[
D^{\bbf}(\p) = \argmax \left\{ \langle \bbf-\p, \x \rangle \,:\, \x \in \R^n_{\geq 0} \text{ and } \langle \p, \x \rangle \leq \beta(\bbf) \right\}.
\]
We call an allocation $\{\x^{\bbf}\}_{\bbf \in \mathcal{B}}$ envy-free, if at given prices $\p$ bidders receive bundles in their respective demand sets, i.e., $x^{\bbf} \in D^{\bbf}(\p)\, \forall \bbf \in \mathcal{B}$.

The auctioneer must provide enough goods such that the selected allocation $\{\x^{\bbf}\}_{\bbf \in \mathcal{B}}$ can be realized. The cost of providing the required amount of goods is $\psi(\sum_{\bbf \in \mathcal{B}} \x^{\bbf})$. We assume $\psi: \R^n_{\geq 0} \rightarrow \R \cup \{+\infty\}$ to be a lower semi-continuous convex function with $\psi(\mathbf{0})=0$ which is monotonically increasing: if $\x \leq \y$ coordinate-wise, then $\psi(\x) \leq \psi(\y)$. The auctioneer's revenue for setting prices to $\p$, assigning bundles $\x^{\bbf}$ to the bidders and receiving their payment is then given by
\[
\sum_{\bbf \in \mathcal{B}} \langle \p, \x^{\bbf} \rangle - \psi\left(\sum_{\bbf \in \mathcal{B}} \x^{\bbf} \right).
\]

The budget-constrained Product-Mix auction tries to balance the interests of the seller and the bidders: the auctioneer's revenue should be maximized among all envy-free allocations:
\begin{align}\label{eqn:rev_max_basic}
\max & \sum_{\bbf \in \mathcal{B}} \langle \p, \x^{\bbf} \rangle - \psi\left(\sum_{\bbf \in \mathcal{B}} \x^{\bbf} \right) \\ \nonumber
\text{s.t.} & \, \x^{\bbf} \in D^{\bbf}(\p) \, \forall \bbf \in \mathcal{B} \\ \nonumber
& \, \p \in \R^n_{> 0}.
\end{align}

\section{Contribution}
Due to the non-linear constraints $\x^{\bbf} \in D^{\bbf}(\p)$, we may not expect existence of an efficient polynomial-time algorithm that finds a global optimal solution to Problem (\ref{eqn:rev_max_basic}). Paul Klemperer and DotEcon proposed an algorithm with good running time behavior when the number of different goods is relatively small (they assumed $n=3$ in their practical application). We state and analyze this algorithm formally. Our main contribution is a proof of correctness for their algorithm.

\section{The Algorithm}
The basic idea of the algorithm is as follows. Suppose we fix a price vector $\p \in \R^n_{> 0}$. In order to determine a revenue maximizing allocation at this specific price, we must solve
\begin{align}\label{eqn:fixed_p}
r(\p) = \max & \sum_{\bbf \in \mathcal{B}} \langle \p, \x^{\bbf} \rangle - \psi\left(\sum_{\bbf \in \mathcal{B}} \x^{\bbf} \right) \\ \nonumber
\text{s.t.} & \, \x^{\bbf} \in D^{\bbf}(\p) \, \forall \bbf \in \mathcal{B}.
\end{align}
Since the demand sets $D^{\bbf}(\p)$ are convex and compact and the objective function is concave, an optimal solution to this problem exists and can be determined efficiently. Problem (\ref{eqn:rev_max_basic}) can now be rewritten as $\max_{\p \in \R^n_{> 0}} r(\p)$. Our strategy is to determine a finite list of prices $\p^1,\dots,\p^m$ among which a maximizer of $r(\p)$ is guaranteed to be found. By computing $r(\p^j)$ for all these prices, we find a global optimal solution to Problem (\ref{eqn:rev_max_basic}). The \emph{Filtered Prices} Algorithm \ref{alg:filtered_prices} iterates over all tuples $(\bbf^1,\dots,\bbf^n) \in \mathcal{B}^n$ of bids and permutations $\sigma$ and calls the generate\_candidate Algorithm \ref{alg:generate_candidate} for all these inputs, which either returns a candidate price, or returns INFEASIBLE. Whenever the latter occurs, this allows us to ignore certain permutations as inputs to Algorithm \ref{alg:generate_candidate}, reducing the number of iterations of the Filtered Prices Algorithm. Our main result is that the list of prices generated by Algorithm \ref{alg:filtered_prices} contains a globally optimal solution to Problem (\ref{eqn:rev_max_basic}).

\begin{theorem}\label{thm:main}
Let $P$ be the set of prices returned by Algorithm \ref{alg:filtered_prices}. Then $P$ contains a price $\p^*$ which is an optimal solution to Problem (\ref{eqn:rev_max_basic}). Consequently, by evaluating $r(\p)$ for all $\p \in P$, we solve Problem (\ref{eqn:rev_max_basic}) globally.
\end{theorem}

Algorithm \ref{alg:filtered_prices} is based on the insight that we can restrict our search for possible maximizing prices to points where bidders are indifferent between multiple goods. For the ease of notation, we introduce a zeroth dummy good, corresponding to the empty bundle. We add an additional zeroth entry to prices $\p \in \R^n_{> 0}$ by setting $p_0 = 1$ and do the same for all bids $\bbf \in \mathcal{B}$ by also setting $b_0 = 1$. For $i \in \{1,\dots,n\}$, we write $\e_i$ for the $i$-th standard vector and set $\e_0 = \mathbf{0} \in \R^n$.
\begin{lemma}\label{lem:demand_set}
	Let $\p \in \R^n_{> 0}$ be a price vector and $\bbf \in \mathcal{B}$ a bid. Then
	\[
	D^{\bbf}(\p) = \conv \left\{ (\beta(\bbf)/p_i)\e_i \,:\, i \in \{0,\dots,n\} \text{ and } b_i p_j \geq b_j p_i \, \forall j\in \{0,\dots,n\} \right\}.
	\]
\end{lemma}
The proof directly follows from the definition of the demand set and is omitted. We say that bidder $\bbf$ demands good $i \in \{0,\dots,n\}$ at prices $\p$, if $(\beta(\bbf)/p_i)\e_i \in D^{\bbf}(\p)$. Equivalently, good $i$ is demanded if and only if $b_i p_j \geq b_j p_i$ for all $j$. Given two distinct goods $i,j \in \{0,\dots,n\}$, we consider the region in price space where a bidder $\bbf$ is indifferent between these two goods, i.e., demands both of them. 
\begin{definition}
	A pair $(\bbf,I)$, where $\bbf$ is a bid and $I=\{i,j\} \subset \{0,\dots,n\}$ with $|I|=2$ is an unordered pair of indices, induces the \emph{indifference hyperplane}
	\[
	H(\bbf, I) = \{\p \in \R^n_{> 0} \,:\, b_i p_j = b_j p_i \text{ and } b_i p_k \geq b_k p_i \,\forall k=0,\dots,n\}.
	\]
	
	We say that a collection $(\bbf^1,I^1), \dots, (\bbf^m,I^m)$ of indifference hyperplanes is \emph{linearly independent}, if their defining linear equations $b_i p_j = b_j p_i$ are linearly independent (where we interpret $p_0$ as the constant $1$, not as  a variable).
\end{definition}
\begin{figure}[t]
	\centering
	\includegraphics[width=0.5\textwidth]{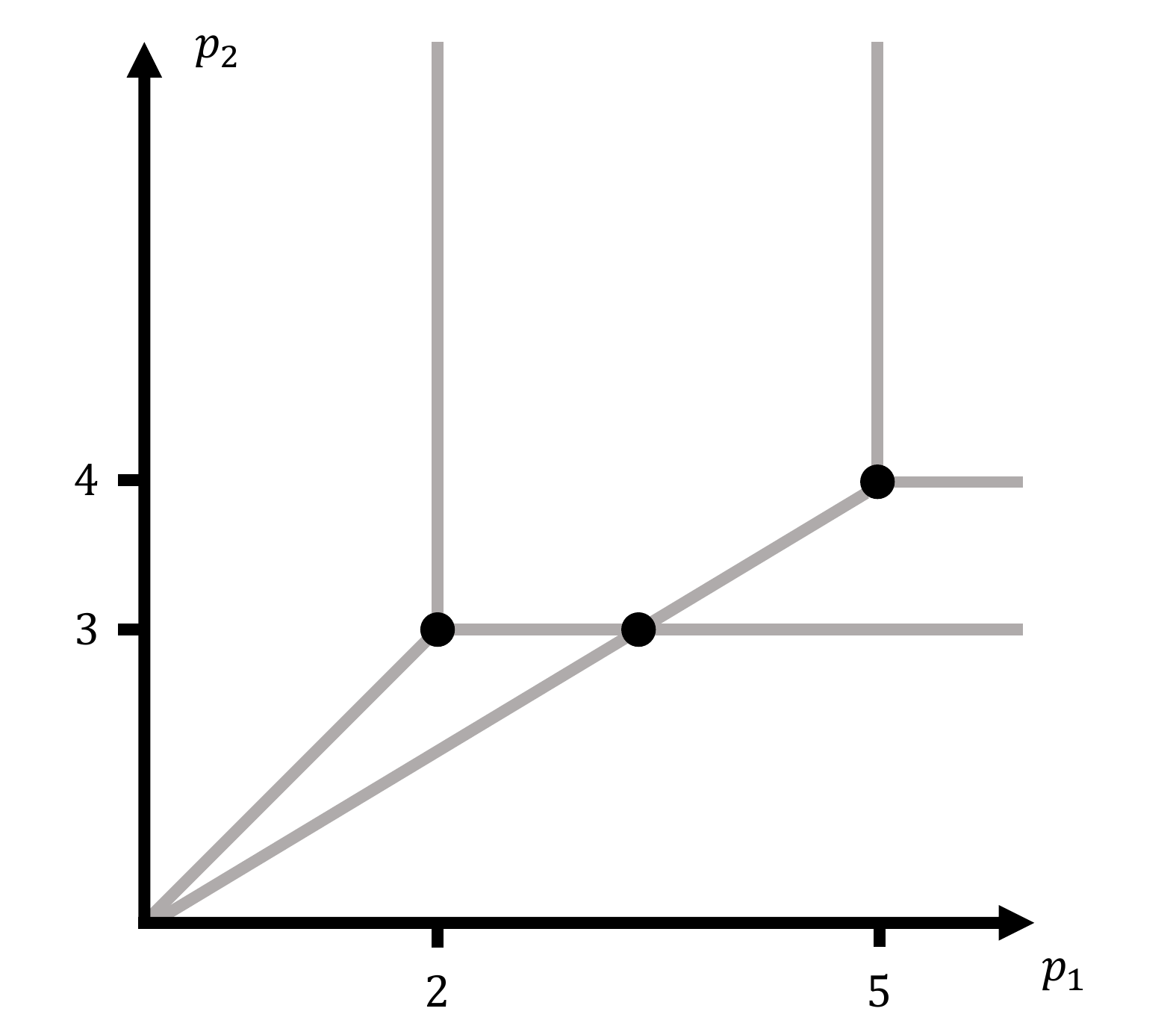}
	\caption{Indifference hyperplanes (gray) of bids at $(2,3)$ and $(5,4)$. Candidates for optimal prices are indicated by black dots.}
	\label{fig:price_space}
\end{figure}
The following proposition states that we can restrict our search for candidate prices to intersections of $n$ linearly independent indifference hyperplanes (see Figure \ref{fig:price_space}).
\begin{proposition}\label{prop:price_at_intersection}
	There is an optimal price vector $\p^*$ for Problem \ref{eqn:rev_max_basic} which is the (necessarily unique) intersection point of $n$ linearly independent indifference hyperplanes: $\{\p^*\} = \bigcap_{i=1}^n H(\bbf^i,I^i)$.
\end{proposition}

A proof can be found in the Appendix. With Proposition \ref{prop:price_at_intersection} at hand, we will show that for each intersection point of $n$ linearly independent indifference hyperplanes, there is an input for the generate\_candidate Algorithm \ref{alg:generate_candidate} such that this intersection point is returned.

The next lemma provides an insight into the structure of linearly independent indifference hyperplanes.
\begin{algorithm}[t]
	\label{alg:generate_candidate}
	\KwIn{Bids $(\bbf^1,\dots,\bbf^n) \in \mathcal{B}^n$, permutation $\sigma \in S_n$.}
	\KwOut{Price vector $\p$, or INFEASIBLE}
	Set $C_1^1, \dots, C^n_1 = 1$\;
	\For{$k=1,\dots,n$}{
		Set $p_{\sigma(k)} = C_{k}^k b^k_{\sigma(k)}$\;
		\If{there exists $m < k$ with $b^m_{\sigma(m)}p_{\sigma(k)} < b^m_{\sigma(k)}p_{\sigma(m)}$ \label{algstep:is_feasible}} {
			\textbf{return} INFEASIBLE \label{algstep:returninf}\;
		}
		\For{$l=k+1,\dots,n$ \label{algstep:update_C_start}}{
			\eIf{$p_{\sigma(k)} < C^l_{k} b^l_{\sigma(k)}$}{
				Set $C^l_{k+1} = p_{\sigma(k)}/b^l_{\sigma(k)}$\;
			}{
				Set $C^l_{k+1} = C^l_{k}$\; \label{algstep:update_C_end}
			}
		}
	}
	\textbf{return} $\p$\;
	\caption{generate\_candidate}
\end{algorithm}

\begin{lemma}\label{lem:hyperplane_ordering}
	Let $(\bbf^1,I^1),\dots,(\bbf^n,I^n)$ be $n$ linearly independent indifference hyperplanes intersecting in a price $\p \in \R^n_{> 0}$. Let $\mathcal{I} \subset \{0,\dots,n\}$ with $0 \in \mathcal{I}$ and $|\mathcal{I}| < n+1$. Then there exists some $l$ with $|I^l \cap \mathcal{I}| = 1$.
\end{lemma}
\begin{proof}
	Assume this is not the case. Then for any indifference hyperplane $(\bbf^l,I^l)$, either both indices of the variables of the defining linear equation are in $\mathcal{I}$, or none of them. So by linear independence, there must be equations with both variables in $\mathcal{I}^c$ and those determine the values $p_i$ for $i \in \mathcal{I}^c$ uniquely. These equations are all of the form $b^l_i p_j = b^l_j p_i$ for $i,j \neq 0$ and some $l$. One and thus the unique solution to these equations is $p_i = 0$ for $i \in \mathcal{I}^c$, contradicting the positivity of $\p$.
\end{proof}

The variables $C^l_k$ in Algorithm \ref{alg:generate_candidate} have a simple intuitive meaning, as described in the next lemma. The proof only consists of inspecting steps \ref{algstep:update_C_start}-\ref{algstep:update_C_end} of the algorithm, so we omit it.

\begin{lemma}\label{lem:factor_values}
 For $1 \leq k \leq l \leq n$ we have that $C^l_{k} = \min_{m=0,\dots,k-1} p_{\sigma(m)}/b^l_{\sigma(m)}$, where we set $\sigma(0)=0$.
\end{lemma}

We are now in a position to prove our main result: if a price vector $\p$ is an intersection of $n$ linearly independent indifference hyperplanes, then there is an input to the generate\_candidate Algorithm \ref{alg:generate_candidate} such that $\p$ is returned.
\begin{proposition}\label{prop:main}
	For each price vector $\p \in \R^n_{> 0}$ that is the intersection point of $n$ linearly independent indifference hyperplanes $(\bbf^1,I^1), \dots, (\bbf^n,I^n)$, there is a reordering of the bids $(\bbf^1,\dots,\bbf^n)$ and a permutation $\sigma \in S_n$ such that Algorithm $1$ returns $\p$.
\end{proposition}

\begin{proof}
	By Lemma \ref{lem:hyperplane_ordering}, we may assume without loss of generality that the indifference hyperplanes are ordered such that the following holds: $|\{0\} \cap I^1| = 1$ and for all $k > 1$, $|I^k \cap \bigcup_{l=1}^{k-1} I^l|=1$. Set $\{\sigma(1)\} = I^1 \setminus \{0\}$ and $\{\sigma(k)\} = I^k \setminus \bigcup_{l=1}^{k-1} I^l$.
	
	Let be $\bar \p$ the price returned by Algorithm \ref{alg:generate_candidate}. We show by induction that in the $k$-th iteration the algorithm sets $\bar p_{\sigma(k)} = p_{\sigma(k)}$.
	
	For $k=1$, the algorithm sets $\bar p_{\sigma(1)} = b^1_{\sigma(1)}$. The linear equation defined by $(\bbf^1,I^1)$ is $b^1_0 p_{\sigma(1)} = b^1_{\sigma(1)} p_0$. Since $p_0 = b^1_0 = 1$, it follows that also $p_{\sigma(1)} = b^1_{\sigma(1)}$.
	
	Now let $k>1$. The defining equation of $(\bbf^k, I^k)$ is $b^k_{\sigma(k)} p_{\sigma(l)} = b^k_{\sigma(l)} p_{\sigma(k)}$ for some $l < k$, so $p_{\sigma(k)} = p_{\sigma(l)} b^k_{\sigma(k)}/b^k_{\sigma(l)}$. By Lemma \ref{lem:factor_values}, \[C^k_k = \min_{m=0,\dots,k-1} \bar p_{\sigma(m)}/b^k_{\sigma(m)} = \min_{m=0,\dots,k-1} p_{\sigma(m)}/b^k_{\sigma(m)}.\] Since $\p$ is contained in the indifference hyperplane generated by $(\bbf^k,I^k)$, we have that $b^k_{\sigma(l)}/p_{\sigma(l)}$ is maximal among all quotients $b^k_m/p_{\sigma(m)}$. Consequently, $C^k_k = p_{\sigma(l)}/b^k_{\sigma(l)}$ which implies that the algorithm also sets $\bar p_{\sigma(k)} = p_{\sigma(l)} b^k_{\sigma(k)}/b^k_{\sigma(l)}$. Note that the algorithm does not return INFEASIBLE: Since $p \in H(I^m,\mathbf{b}^m)$ for all $m < k$ and $\sigma(m) \in I^m$, we have that $b^m_{\sigma(m)} p_{\sigma(k)} \geq b^m_{\sigma(k)} p_{\sigma(m)}$. This finishes the proof.
\end{proof}

We can now prove our main result.

\begin{proof}[Proof of Theorem \ref{thm:main}]
	By Proposition \ref{prop:price_at_intersection}, there is an optimal price vector at the intersection of $n$ linearly independent indifference hyperplanes, and by Proposition \ref{prop:main}, there is some input $(\bbf^1,\dots,\bbf^n) \in \mathcal{B}^n$ and $\sigma \in S_n$ for Algorithm \ref{alg:generate_candidate} returning that price vector. Note that this input permutation $\sigma \in S_n$ has not been removed from the set $S$ before: otherwise, there would exist a permutation $\sigma'$ with $\sigma'(l) = \sigma(l) \, \forall l=1,\dots,k$ for some $k$ which caused Algorithm \ref{alg:generate_candidate} to return INFEASIBLE in iteration $k$. This however would cause Algorithm \ref{alg:generate_candidate} to also return INFEASIBLE for the permutation $\sigma$.
\end{proof}

\begin{algorithm}[t]
	\label{alg:filtered_prices}
	\KwIn{Set of bids $\mathcal{B}$}
	\KwOut{Set of candidate price vectors $P$}
	Set $P = \emptyset$\;
	\For{$(\bbf^1,\dots,\bbf^n)\in \mathcal{B}^n$}{
		Set $S = S_n$, the set of all permutations\;
		\For{$\sigma \in S$}{
		Execute generate\_candidate with input $(\bbf^1,\dots,\bbf^n)$ and $\sigma$\;
		\eIf{INFEASIBLE}{
			Let $k$ be the iteration in which INFEASIBLE was returned\;
			Remove all $\sigma'$ from $S$ where $\sigma'(l) = \sigma(l)\, \forall l \leq k$\;
		}{
		Set $P = P \cup \{\p\}$\;
		}
	}
	}
	\textbf{return} $P$\;
	\caption{filtered\_prices}
\end{algorithm}

\begin{remark}
	The number of all combinations of $n$-tuples of bids and all permutations is equal to $|\mathcal{B}|^n n!$. While this grows very fast in the number of goods $n$, it is polynomial in the number of bids when the number of goods is fixed. We also may expect that due to the feasibility check in steps \ref{algstep:is_feasible}-\ref{algstep:returninf} of the generate\_candidate algorithm, the practical runtime is reduced significantly: whenever INFEASIBLE is returned in step $k$, we remove $(n-k)!$ permutations from $S$. So in practice we may expect the algorithm to iterate over much less than $|\mathcal{B}|^n n!$ combinations of inputs, and consequently to return the list of candidate prices reasonably fast when the number of different goods is not too large.
\end{remark}
\section{Conclusion}
We have studied an algorithm proposed by Paul Klemperer and DotEcon for finding revenue maximizing prizes in the arctic Product-Mix auction and formally proved correctness of their approach. The algorithm is sufficiently fast when the number of different goods is small. Further research questions include a formal proof of the hardness of the problem, as well as the design of efficient approximation algorithms.

\section*{Acknowledgements}
I would like to thank Edwin Lock from the University of Oxford for his very valuable comments and suggestions.

\newpage

\appendix
\section{Proof of Proposition \ref{prop:price_at_intersection}}
We prove Proposition \ref{prop:price_at_intersection} with the help of two lemmas. Let us define the set $G^{\bbf}(\p)$ of all goods demanded by bidder $\bbf$ at prices $\p$:
\begin{align*}
G^{\bbf}(\p) = \left\{i \in \{0,\dots,n\} \,:\, (\beta(\bbf)/ p_i)\e_i \in D^{\bbf}(\p)\right\} =\\= \left\{ i \in \{0,\dots,n\} \,:\, b_i/p_i \text{ maximal} \right\}.
\end{align*}
The first lemma states that we only need to consider prices at which every good is demanded by some bidder.
\begin{lemma}\label{lem:all_goods_demanded}
	Let $\bar \p \in \R^n_{> 0}$ be arbitrary. Then there exists $\p \in \R^n_{> 0}$ with \\$\bigcup_{\bbf \in \mathcal{B}} G^{\bbf}(\p) \supseteq \{1,\dots,n\}$ and $r(\p) \geq r(\bar \p)$.
\end{lemma}
\begin{proof}
	Suppose that $k \not\in \bigcup_{\bbf \in \mathcal{B}} G^{\bbf}(\bar \p)$. Then for all bids $\bbf$ and all $i \in G^{\bbf}(\bar \p)$ we have $\bar p_k > (b_k/b_i)\bar p_i$. Define $\p \in \R^n_{> 0}$ by $p_i = \bar p_i$ for $i \neq k$ and $p_k = \max_{\bbf \in \mathcal{B}} \max_{i \in G^{\bbf}(\p)} (b_k/b_i) \bar p_i$. Then we have $b_i/p_i  \geq b_k/ p_k$ for all $\bbf$ and all $i \in G^{\bbf}(\bar \p)$, so $G^{\bbf}(\p) \supseteq G^{\bbf}(\bar \p)$. Moreover, for maximizing bids $\bbf$ and indices $i \in G^{\bbf}(\bar \p)$ in the definition of $p_k$ we have $b_i/ p_i = b_k/ p_k$, so $k \in G^{\bbf}(\p)$. It can easily be seen from Lemma \ref{lem:demand_set} that every feasible solution for $r(\bar \p)$ is also feasible for $r(\p)$. Moreover, any assigned bundle $\x^{\bbf}$ in an optimal solution of $r(\bar \p)$ has $x^{\bbf}_k = 0$, so $\langle \p, \x^{\bbf} \rangle = \langle \bar \p,  \x^{\bbf} \rangle$. This implies $r(\p) \geq r(\bar \p)$.
\end{proof}
The second lemma says that if we increase prices such that the set of demanded goods get larger, the revenue does no decrease.
\begin{lemma}\label{lem:price_order}
	Let $\bar \p, \p \in \R^n_{> 0}$ such that $\p \geq \bar \p$ coefficient-wise and $G^{\bbf}(\p) \supseteq G^{\bbf}(\bar \p)$ for all $\bbf \in \mathcal{B}$. Then $r(\p) \geq r(\bar \p)$.
\end{lemma}
\begin{proof}
	Let $\{\bar \x^{\bbf}\}_{\bbf \in \mathcal{B}}$ be an optimal allocation for $r(\bar \p)$. Then each $\bar \x^{\bbf}$ can be written as a convex combination $\bar \x^{\bbf} = \sum_{i \in G^{\bbf}(\bar \p)} \lambda_i (\beta(\bbf)/\bar p_i)\e_i$. Since $G^{\bbf}(\p) \supseteq G^{\bbf}(\bar \p)$, the allocation $\{\x^{\bbf}\}_{\bbf \in \mathcal{B}}$ defined by $\x^{\bbf}  = \sum_{i \in G^{\bbf}(\bar \p)} \lambda_i (\beta(\bbf)/ p_i)\e_i$ is feasible for $r(\p)$, and $\langle  \bar \p, \bar \x^{\bbf} \rangle  = \langle \p, \x^{\bbf} \rangle$ for all $\bbf$. Since $\x^{\bbf} \leq \bar \x^{\bbf}$ for every bid $\bbf$, we have by monotonicity that $\psi(\sum_{\bbf \in \mathcal{B}} \x^{\bbf}) \leq \psi(\sum_{\bbf \in \mathcal{B}} \bar \x^{\bbf})$. It follows that $r(\p) \geq r(\bar \p)$.
\end{proof}
\begin{proof}[Proof of Proposition \ref{prop:price_at_intersection}]
	Let $\bar \p$ be an arbitrary price vector. We demonstrate that there is a price vector $\p^*$ which is the unique intersection point of $n$ linearly independent indifference hyperplanes with $r(\p^*) \geq r(\bar \p)$. This proves the statement. By Lemma \ref{lem:all_goods_demanded} we may assume that $\bigcup_{\bbf \in \mathcal{B}} G^{\bbf}(\bar \p) \supseteq \{1,\dots,n\}$. Define the set $P$ of not too small prices at which bidders demand a superset of goods they demand at $\bar \p$:
	\[
	P = \left\{\p \,:\, G^{\bbf}(\p) \supseteq G^{\bbf}(\bar \p) \, \forall \bbf \in \mathcal{B} \text{ and } p_i \geq \bar p_i /2 \, \forall i \right\}.
	\]
	$P$ is convex as it can be written as
	\[
	P = \left\{\p \, :\, b_i p_j \geq b_j p_i \,\forall \bbf \in B \,\forall i \in G^{\bbf}(\bar \p) \,\forall j\in\{0,...,n\} \text{ and } p_i \geq \bar p_i/2 \, \forall i \right\}.
	\]
	Since each good is $i$ is demanded by some bid $\bbf$ at $\bar \p$, we have that $p_i = b_0 p_i \leq b_i p_0 = b_i$ for some $\bbf$ and all $\p \in P$, so $P$ is compact. Furthermore, it can be checked that for $\p^1, \p^2 \in P$, their coefficient-wise maximum also lies in $P$. It follows that $\p^* \in \R^n_{> 0}$ defined by $p^*_i = \sup_{\p \in P} p_i \geq \bar p_i$ is contained in $P$. Since $\p^*$ is the unique solution to the linear program $\max_{\p \in P}\sum_{i=1}^n p_i$, $n$ linearly independent of the defining inequalities of $P$ are equalities at $\p^*$. As $\p^* \geq \bar \p$, none of the inequalities $p^*_i \geq \bar p_i/2$ can be tight, so there are $n$ linearly independent equalities of the form $b_j p^*_i = b_i p^*_j$ for some $i \neq j$. By definition of $P$, for each of the equalities $b_j p^*_i = b_i p^*_j$, we have that $b_i p^*_k \geq b_k p^*_i$ for all $k$, so $\p^*  \in H(\bbf,\{i,j\})$. It follows that $\p^*$ is the unique intersection point of the $n$ linearly independent indifference hyperplanes corresponding to these equations. By definition of $P$ we have that $G^{\bbf}(\p^*) \supseteq G^{\bbf}(\bar \p)$, so $r(\p^*) \geq r(\bar \p)$ by Lemma \ref{lem:price_order}.
\end{proof}

\end{document}